%% file: paper.tex
\documentclass[conference]{IEEEtran}
\IEEEoverridecommandlockouts
\usepackage{textcomp}
\usepackage{amsmath,amssymb,amsfonts,amsthm}
\usepackage{mathtools}
\usepackage{graphicx}
\usepackage[vlined,ruled,linesnumbered]{algorithm2e}
\SetKw{Continue}{continue}
\SetKw{Break}{break}
\usepackage{algorithmic}
\usepackage{hyperref}
\usepackage{cleveref}
\usepackage{mdwlist}
\usepackage{color}
\usepackage[usenames,dvipsnames,svgnames,table]{xcolor}
\usepackage{enumitem}
\setlist[itemize]{nosep}  
\setlist[enumerate]{nosep}  
\usepackage{kotex}
\usepackage[subrefformat=parens]{subcaption}
\usepackage{footnote}
\usepackage{multirow}
\usepackage{booktabs}
\usepackage{makecell}
\usepackage{bm}
\usepackage{adjustbox}
\usepackage{pdfrender}
\usepackage[most]{tcolorbox}
\usepackage{diagbox}

\input{dfn}

\def\BibTeX{{\rm B\kern-.05em{\sc i\kern-.025em b}\kern-.08em
    T\kern-.1667em\lower.7ex\hbox{E}\kern-.125emX}}
\begin{document}
	
\title{\huge \method: Principled and Scalable Recommendation Justification}

\author{\IEEEauthorblockN{Namyong Park\textsuperscript{1}$^{*}$, Andrey Kan\textsuperscript{2}, 
Christos Faloutsos\textsuperscript{1}$^{*}$, Xin Luna Dong\textsuperscript{2}}
\thanks{${}^*$Work performed while at Amazon.}
\IEEEauthorblockA{\textsuperscript{1}Carnegie Mellon University, \textsuperscript{2}Amazon\\
\{namyongp,christos\}@cs.cmu.edu,~\{avkan,lunadong\}@amazon.com}}

\maketitle

\begin{abstract}
Online recommendation is an essential functionality across a variety of services, including e-commerce and video streaming, where items to buy, watch, or read are suggested to users. Justifying recommendations, i.e., explaining why a user might like the recommended item, has been shown to improve user satisfaction and persuasiveness of the recommendation. In this paper, we develop a method for generating post-hoc \es that can be applied to the output of any recommendation algorithm. Existing post-hoc methods are often limited in providing diverse \es, 
as they either use only one of many available types of input data, or rely on the predefined templates.
We address these limitations of earlier approaches by developing \method, a method for producing concise and diverse \es. 
\method is a recommendation model-agnostic method that generates diverse \es based on various types of product and user data (e.g., purchase history and product attributes). The challenge of jointly processing multiple types of data is addressed by designing a principled graph-based approach for \e generation. In addition to theoretical analysis, we present an extensive evaluation on synthetic and real-world data. Our results show that \method satisfies desirable properties of \es, 
and efficiently produces effective \es, 
matching user preferences up to 20\% more accurately than baselines.
\end{abstract}


\section{Introduction}
\label{sec:intro}
\input{010introduction}

\section{Justifying Recommendations}
\label{sec:methods}
\input{020methods}

\section{Evaluation Using Axioms}
\label{sec:axiomeval}
\input{030axiomeval}

\section{Evaluation Using Real-World Data}
\label{sec:realeval}
\input{040realeval}

\section{Related Work}
\label{sec:related}
\input{050related}

\vspace{-0.5em}
\section{Conclusion}
\label{sec:concl}
\input{060conclusion}

\vspace{-0.5em}
\bibliographystyle{IEEEtran}

\appendix
\input{070appendix}

\end{document}

%% file: dfn.tex
\newtheorem{theorem}{Theorem}

\newtheorem{problem}{Problem}

\definecolor{Gray}{gray}{0.9}
\definecolor{LightCyan}{rgb}{0.88,1,1}
\newcommand{\cyancell}[1]{\cellcolor{LightCyan}#1}

\newcommand{\highlight}[1]{{\textcolor{blue}{#1}}}

\newcommand{\tablefont}{}

\newcommand*{\belowrulesepcolor}[1]{
	\noalign{
		\kern-\belowrulesep
		\begingroup
		\color{#1}
		\hrule height\belowrulesep
		\endgroup
	}
}
\newcommand*{\aboverulesepcolor}[1]{
	\noalign{
		\begingroup
		\color{#1}
		\hrule height\aboverulesep
		\endgroup
		\kern-\aboverulesep
	}
}

\newcommand*{\boldcheckmark}{
	\textpdfrender{
		TextRenderingMode=FillStroke,
		LineWidth=.7pt, 
	}{\checkmark}
}

\newcommand{\e}{justification\xspace}
\newcommand{\es}{justifications\xspace}
\newcommand{\E}{Justification\xspace}
\newcommand{\ES}{Justifications\xspace}

\newcommand{\pg}{product graph\xspace}
\newcommand{\pgs}{product graphs\xspace}

\newcommand{\PG}{Product Graph\xspace}

\newcommand{\attr}{attribute\xspace}
\newcommand{\Attr}{Attribute\xspace}
\newcommand{\attrs}{attributes\xspace}
\newcommand{\Attrs}{Attributes\xspace}

\DeclareMathOperator*{\argmax}{arg\,max}

\newcommand{\setProd}{{\mathcal{P}}} 
\newcommand{\setFeedback}{{\mathcal{Q}}} 
\newcommand{\setAttribute}{{\mathcal{A}}} 
\newcommand{\setData}{{\mathcal{D}}} 
\newcommand{\setTopic}{{\mathcal{T}}}
\newcommand{\numProd}{{L}}
\newcommand{\prodG}{{\mathcal{G}}}
\newcommand{\p}{{\mathrm{Pr}}}
\newcommand{\numSelection}{{50}\xspace}

\newcommand{\movie}{\textsc{movie-pg}\xspace}
\newcommand{\paper}{\textsc{citation-pg}\xspace}
\newcommand{\paperlarge}{\textsc{citation-100m-pg}\xspace}

\newcommand{\method}{\textsc{J-Recs}\xspace}
\newcommand{\pr}{$ \mathrm{PR} $\xspace}
\newcommand{\ppr}{$ \mathrm{PPR} $\xspace}
\newcommand{\har}{$ \mathrm{HAR} $\xspace}

\newcommand{\explod}{$ \mathrm{ExpLOD} $\xspace}
\newcommand{\mpr}{$ \mathrm{MP} $\xspace}
\newcommand{\mpand}{MP-AND\xspace}
\newcommand{\mpor}{MP-OR\xspace}

%% file: 010introduction.tex
Recommender systems have a profound and ever increasing impact on how online users make purchase decisions, consume various types of content, and engage with the service. While recommender systems have seen significant progress in terms of recommendation accuracy, algorithms widely used in practice are mostly black boxes. This includes recommenders based on the latent factor models such as matrix factorization~\cite{DBLP:journals/computer/KorenBV09,DBLP:conf/sdm/ZhangWFM06}, as well as some deep learning-based recommenders~\cite{DBLP:conf/kdd/WangWY15,DBLP:conf/www/Fan0LHZTY19}. Such systems can be limited in their ability to justify recommendations.

Justification refers to explaining why a user might like the recommended item~\cite{biran2017explanation}. 
In other words, while recommendations suggest users \textit{what} they might like, 
justifications reveal \textit{why} the recommended item might match their preferences.
For instance, a list of recommended products can be supplemented with a justification that ``these items are similar to what you recently purchased.'' Several studies have shown that justifications can improve user satisfaction~\cite{herlocker2000explaining}, increase the persuasiveness and reliability of recommendations~\cite{DBLP:conf/icde/TintarevM07,DBLP:reference/sp/TintarevM15}, and help users make more accurate and efficient decisions~\cite{bilgic2005explaining}.

In this paper, we focus on post-hoc justification of recommendations. In post-hoc approaches, recommendations and \es are decoupled from each other; that is, \es are generated after the recommendation has been given. The main advantage of generating \es post-hoc is that post-hoc methods can be easily applied to different types of recommendation algorithms (thus \textit{recommendation model-agnostic}), which allows a greater freedom in the design of explanations~\cite{DBLP:conf/iui/VigSR09}.

Existing post-hoc methods typically select \es from predefined templates~\cite{DBLP:journals/ftir/ZhangC20},  such as ``your neighbors' rating for this item is ...''~\cite{herlocker2000explaining}, or they provide \es based on only one type of data, such as keywords~\cite{bilgic2005explaining},
although many types of data are often available.
While these methods have been shown to produce concise \es, they are limited in their ability to provide diverse \es. 
Moreover, some of these methods generate \es in a non-personalized manner~\cite{DBLP:conf/eacl/ZhouLWDHX17}, 
while other post-hoc methods require labeled ground truth data to train a \e model~\cite{DBLP:conf/emnlp/NiLM19}, thus posing an additional hurdle.

In summary, major challenges of generating post-hoc \es are in handling heterogeneous data 
(e.g., user purchase history, product attributes and reviews) 
to generate flexible and diversified justifications 
without the need for manually labeled data,
while enabling that the \e diversity can be increased without changing the underlying algorithm.
We address these challenges by proposing a novel principled graph-based method called \method. 
We use the graph to represent heterogeneous data that can be leveraged for \es. 
Moreover, the graph-based representation allows us to generate \es personalized with respect to both the user and the recommended item. 
Finally, we derive an objective function that agrees with intuition and leads to concise and diverse \es.
This paper makes the following contributions.
\begin{itemize}[leftmargin=1.0em]
	\item \textbf{Problem Formulation.} We present a graph-based formulation of the problem of generating concise and diverse \es given various types of user and product data.
	\item \textbf{Principled Approach.} We develop \method, a principled post-hoc framework to infer \es. \method is guided by a set of principles characterizing desired \es, and does not require manually labeled data.
	\item \textbf{Effectiveness.} We demonstrate that \method satisfies desirable properties of \es, and show the effectiveness of \method in experiments on real-world data (\Cref{fig:exp:crown:prefretrieval}).
	\item \textbf{Scalability.} Our proposed \method is scalable, and runs in time linear in the size of input data (\Cref{fig:exp:scalability}).
\end{itemize}

\begin{figure}[!t]
\par\vspace{1em}\par
\centering
\makebox[\linewidth][c]{\includegraphics[width=0.95\linewidth]{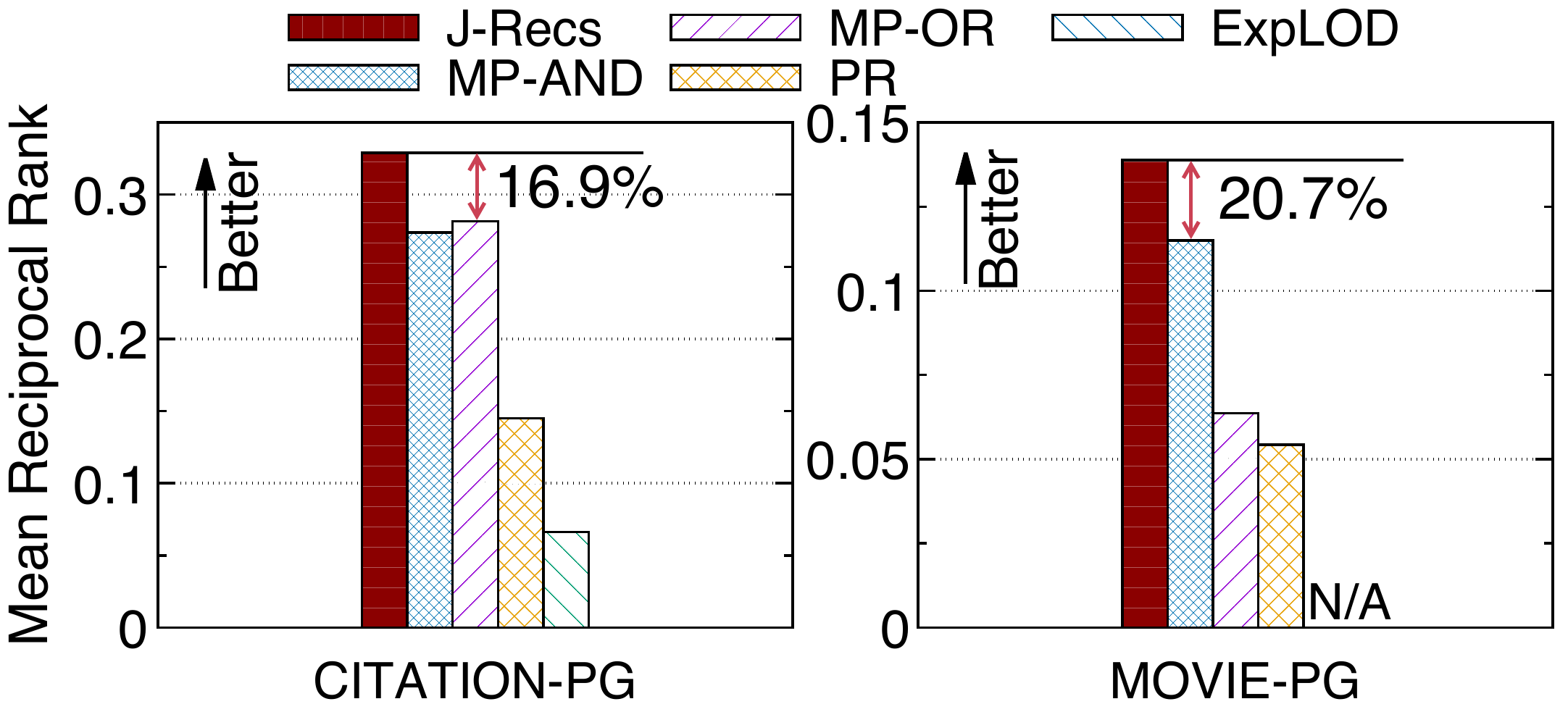}}
\caption{\method generates \es that match user preferences better than existing methods. Higher values are better. See \Cref{sec:realeval:quality} for details. 
}
\label{fig:exp:crown:prefretrieval}
\vspace{-1.0em}
\end{figure}

The rest of the paper is organized as follows. We formulate the problem of graph-based recommendation \e and  present our framework in~\Cref{sec:methods}. Then we provide evaluation results using axioms and real-world data in~\Cref{sec:axiomeval,sec:realeval}, respectively. After discussing related work in~\Cref{sec:related}, we conclude in~\Cref{sec:concl}.

%% file: 020methods.tex
In this section, we first provide the problem statement and define the \pg and \es.
We then describe how \method identifies good \es efficiently.
The symbols used in the paper is given in~\Cref{tab:symbols}.

\vspace{-0.5em}
\subsection{Problem Statement}

Products can be any item, such as movies, audio tracks, and papers, which can be suggested by recommender systems.
Let $ \setProd = \{p_1, p_2, \ldots, p_\numProd \} $ denote the set of all products.
We are given a product $ r_u \in \setProd $,
which is recommended to user $ u $ by an external recommendation algorithm.
We also have a set $ \setFeedback_u \subseteq \setProd $ of products, 
to which user~$ u $ gave positive feedback.
For instance, $ \setFeedback_u $ can be the products user~$ u $ purchased or rated highly.
Let $ q^u_i $ denote the $ i $-th product in $ \setFeedback_u $. 
For simplicity, we may omit subscript $ u $, and use $ Q $, $ r $, and $ q_i $.
Now consider various types of product information,
which we collectively denote by $ \setData $.
Examples of product data include the followings.
\begin{itemize}[leftmargin=1em]
	\item Product details: e.g., flavor, category, color of a product; actors, directors, and genres of a movie
	\item Product keywords: e.g., movie keywords submitted by users 
	\item Product reviews; sentences in the product reviews
	\item Product co-purchase and co-view records
\end{itemize}

Given these input data, our problem is stated as follows:

\begin{tcolorbox}[boxsep=0pt,boxrule=0pt,left=0.1cm,right=0.1cm,top=0.1cm,bottom=0.1cm]
	Given products $ \setProd $, product data $ \setData $, 
	recommended product $ r_u\!\in\!\setProd $, and products $ \setFeedback_u\!\subseteq\!P $ that received positive feedback by user~$ u $,
	efficiently infer \es relevant to the recommendation $ r_u $
	in a way that best reflects the user's preference expressed by $ \setFeedback_u $.
\end{tcolorbox}\vspace{-0.5em}
In the following sections, we further formalize this problem 
by defining the notion of \e, \e score, and the optimization objective.

\subsection{Product Graph and \ES}

A good \e needs to capture the user's preference, while being relevant to the recommended item.
Product data provide useful information that can help with identifying good \es.
A major challenge in effectively employing product data lies in
how to jointly take into account various sources of information in the product data. 

To address this challenge, we need to be able to measure the relevance and similarity between products and product data.
We note that each type of product data provides a signal that lets us identify a set of products that are similar in some specific respect.
For instance, products with the ``chocolate flavor'' are likely to have a similar taste, and
movies that share several keywords tend to have a lot of similarity.
Also, knowledge of similar products can enable us to see the relatedness of seemingly different product attributes.
To make the most of this mutually influential relationship,
we combine all available information into a graph, which we call a \textit{\pg}, 
and find good \es in terms of it.

\begin{figure}[!t]
\centering
\makebox[\linewidth][c]{\includegraphics[width=0.9\linewidth]{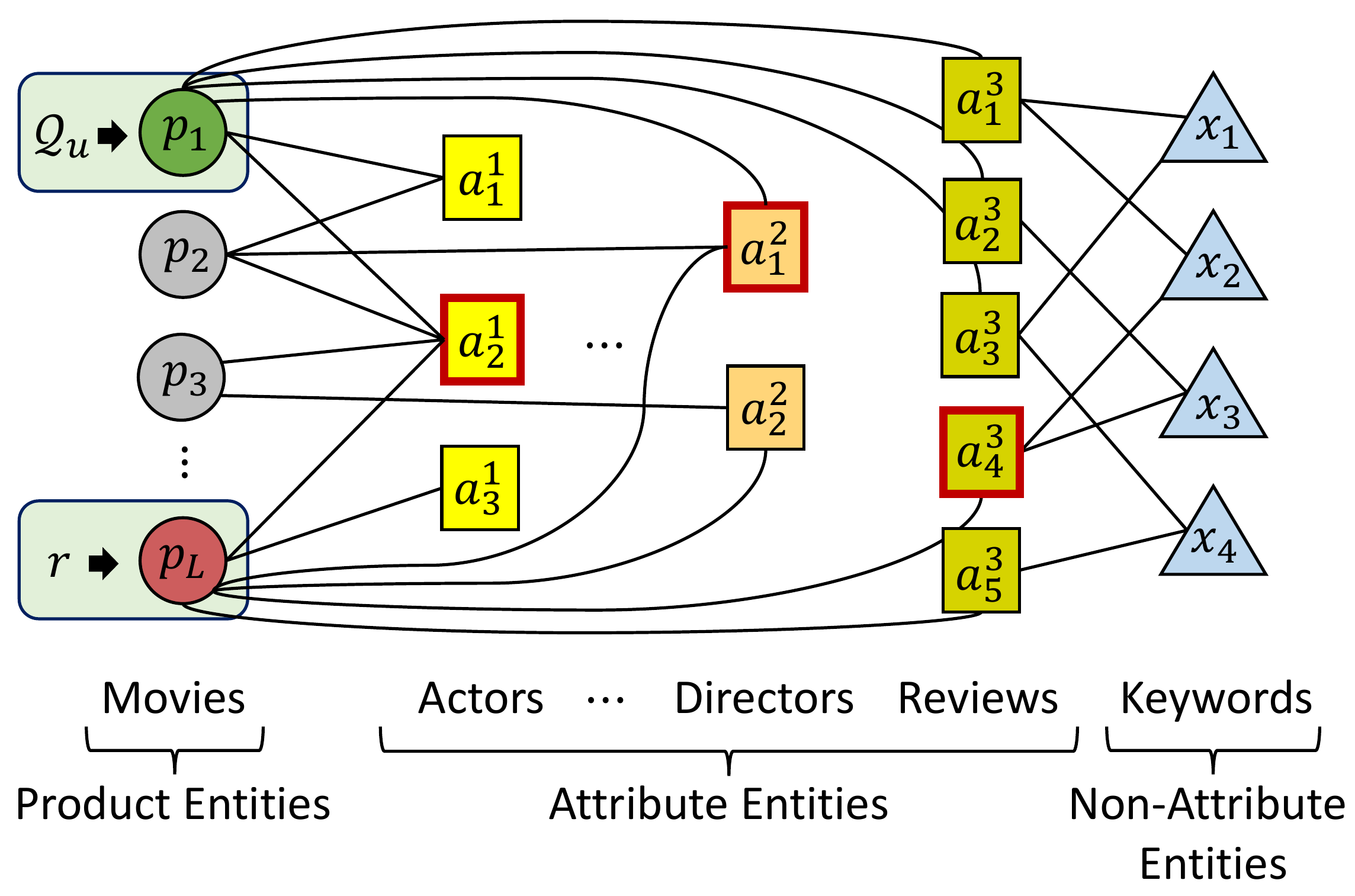}}
\caption{A movie product graph.
	$ \setFeedback_u $ denotes the set of products that received positive feedback from user $ u $, and
	$ r $ denotes the recommended product.
	Boxes with a bold red outline denote the \attr entities that are selected to form \es.}
\label{fig:product_graph}
\vspace{-1.0em}
\end{figure}

\textbf{\PG.} 
We refer to an instance of specific product data as \textit{``\attr entity''} (or \textit{``\attr''} in short) and 
use the term \textit{``\attr type''} to denote a specific type of product data.
For instance, each one of the examples given for the product data (e.g., movie genre, product color, review) represents one \attr type, and
science fiction is an \attr of the movie genre type.
We denote the set of all product \attrs by~$ \setAttribute $.
Products $ \setProd $ and their \attrs $ \setAttribute $ are nodes in a product graph $ \prodG $.
A \pg can also have \textit{non-\attr entities} as nodes,
which are entities not directly connected to products.
Instead, they connect similar product \attrs,
enabling us to identify similar products that do not share the same \attrs.
Examples include facts common to actors, and common review keywords.
\Cref{fig:product_graph} shows an example product graph.

\begin{figure}[!t]
\centering
\makebox[\linewidth][c]{\includegraphics[width=1.05\linewidth]{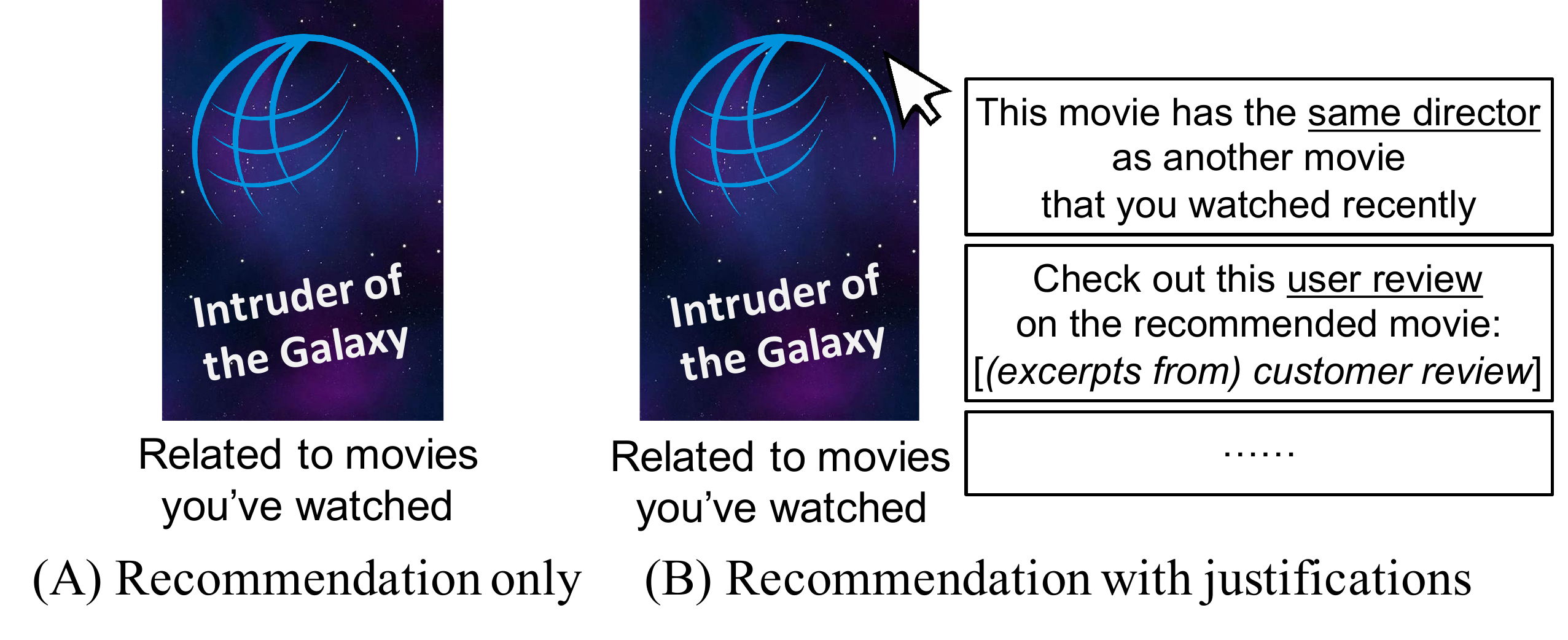}}
\caption{Product graph-based \es enrich the recommendation with relevant and diverse information on why the user might like the recommended item.}
\label{fig:exp:justification}
\vspace{-1.0em}
\end{figure}

\textbf{\E.}
We aim to find a set of relevant product \attrs, such as ``red color'' and ``high efficiency'', and 
produce \es in a format determined by the type of selected \attrs. 
More precisely, given a product graph $ \prodG $ and the recommended item $ r_u $, 
\es for $ r_u $ are a subset $ \setAttribute' \subseteq \setAttribute $ of \attrs
selected among those that are connected to $ r_u $ in $ \prodG $.
In~\Cref{fig:product_graph}, boxes with a bold red outline denote the \attrs that are chosen to form \es.

A product \attr, appropriately chosen in light of user preferences, naturally lends itself to an intuitive and concise \e.
For instance, when the director~James Cameron is chosen for a movie recommendation, 
this translates to the \e that \textit{``We recommend this movie directed by James Cameron, 
who made movies that are similar to other movies you watched.''}
When the user watched some of James Cameron's movies, we can further enrich this \e with that information.
Similarly, when a review gets selected,
this translates to the \e like 
\textit{``Check out this user review, which closely reflects your preference, and explains why you should consider buying this item.''}
\Cref{fig:exp:justification} illustrates how such \es enrich a movie recommendation, making the recommendation more persuasive.
As in these examples, the final \es can be generated 
based on custom rules that consider the \attr type and previous user actions.
Note that, since our \es are based on the product graph, 
the diversity of \es increases as we add more product data to it,
without needing to make changes to the framework.

\subsection{Quantifying the Quality of \ES}

To select good \es tailored to the user and the recommendation, we need to be able to measure the goodness of product \attrs.
Towards this goal, we design \e scores, 
such that higher scores indicate better \es
in consideration of 
the product graph $ \prodG $, products $ \setFeedback_u $, and the recommended item $ r $.
We quantify the \e score by considering two aspects, namely, relevance and diversity.

\subsubsection{\textbf{Relevance Score}}

Intuitively, a good justification should be highly relevant to the recommended product $ r $.
To measure the relevance of a product \attr, 
we consider how probable \attr $ a $ is given the recommended product $ r $, that is,
\begin{align}\label{eq:likelihood}
\p(A=a|R=r)
\end{align}
where $ R $ and $ A $ are random variables denoting a recommended product, and an \attr of the recommended product, respectively.
Consider a random variable $ U $, which denotes the product that matches the preference of user $ u $. 
Given $ \numProd $ products $ p_1,p_2,\ldots,p_\numProd $ in $ \setProd $, 
applying the sum rule of probability, the likelihood \eqref{eq:likelihood} can be expressed as:
\begin{align}\label{eq:sumrule}
\p(A=a|R=r) = \sum_{i=1}^{\numProd} \p(A=a,U=p_i|R=r).
\end{align}
Among the $ \numProd $ terms in~\eqref{eq:sumrule},
we know the user's preferences on the products in $ \setFeedback_u $, and have no information for other products.
So, we assume that the products to which user $ u $ gave no feedback are equally likely given the recommended item.
Specifically, we assume that $ P(A\!=\!a,U\!=\!p_i|R\!=\!r) = P(A\!=\!a,U\!=\!p_j|R\!=\!r) $ for all $ p_i,p_j \notin \setFeedback_u$ and any \attr $ a \in \setAttribute $.
In other words, in measuring the relevance score,
our focus is on the sum of likelihood terms of products in $ \setFeedback_u $: 
\begin{align}\label{eq:sumrule2}
\sum_{i=1}^{|\setFeedback_u|} \p(A=a,U=q_i|R=r).
\end{align}

\begin{table}[!t]
\caption{Table of symbols.}
\centering
\makebox[0.4\textwidth][c]{
	\begin{tabular}{ c | l }
		\toprule
		\textbf{Symbol} & \multicolumn{1}{c}{\textbf{Definition}} \\
		\midrule
		$ \setProd $ & Set of $ \numProd $ products $ (\setProd = \{p_1, p_2, \ldots, p_\numProd \}) $ \\
		$ r_u (r) $ & Product recommended to user $ u $ \\
		$ \setFeedback_u (\setFeedback) $ & Set of products to which user $ u $ gave positive feedback\\
		$ q_i^u (q_i) $ & $ i $-th product in $ \setFeedback_u $ (i.e., $ q_i \in \setFeedback_u $) \\
		$ a $ & Product \attr entity \\
		$ \setAttribute $ & Set of all \attr entities \\
		$ \setAttribute_{(r)} $ & \Attr entities of product $ r $ \\
		$ \setAttribute' $ & Set of selected \attr entities \\
		$ \rho $ & Relative weight of $ r_u $ in comparison to $ q_i $ \\
		$ \lambda_1, \lambda_2 $ & Non-negative weights for the \e diversity \\
		$ B $ & Budget (maximum number of \es to select) \\
		\bottomrule
	\end{tabular}
}
\label{tab:symbols}
\end{table}

We observe that, by the product rule of probability, the following holds for each of the $ |\setFeedback_u| $ terms in \eqref{eq:sumrule2}:
\begin{align}\label{eq:pair}
\p(a,q_i|r) = \underbrace{\p(a|q_i,r)}_{\substack{\text{\Attr}\\\text{relevane}}} \cdot \underbrace{\p(q_i|r)}_{\substack{\text{Feedback}\\\text{relevance}}}
\end{align}
where we have omitted random variables to simplify notations.
In~\eqref{eq:pair}, 
the first term $ \p(a|q_i,r) $ denotes the likelihood of \attr $ a $ given product $ q_i $ and recommended item $ r $, 
which closely matches our goal of finding a \e relevant to both the recommendation and the user's preference.
The second term $ \p(q_i|r) $ represents how probable product $ q_i $ is, given recommendation $ r $.
This term measures the relevance of feedback $ q_i \in \setFeedback_u $ with respect to $ r $,
which enables considering the fact that some of the products in $ \setFeedback_u $ may not be relevant to the current recommendation.
In other words, $ \p(q_i|r) $ acts as a weight for the \attr relevance term.

Let $ R_r(a) $ denote the relevance score of \attr $ a $ as a \e of the recommendation~$ r $,
which we define to be
\begin{align}\label{eq:relevancescore_as_pr}
	R_r(a) = \sum_{i=1}^{|\setFeedback_u|} \p(a|q_i,r) \cdot \p(q_i|r).
\end{align}

\textbf{Modeling \Attr and Feedback Relevance.} 
We model the two relevance terms in \eqref{eq:pair} in terms of the product graph $ \prodG $, 
as it provides rich information on how products and \attrs are related to each other.
Specifically, we consider a random walk using personalized PageRank (PPR)~\cite{DBLP:conf/www/Haveliwala02} over $ \prodG $, and
model the \attr relevance, $ \p(a|q_i,r) $,
by the proximity of $ a $ with respect to $ q_i $ and $ r $ in terms of PPR on $ \prodG $.
Intuitively, given a random walker who travels over $ \prodG $ and returns to $ q_i $ and $ r $ with a fixed probability, 
an \attr which gets visited more times than others is deemed more likely than other \attrs, 
in terms of $ \prodG $ given $ q_i $ and $ r $.
We use the following notation
\begin{align}\label{eq:ppr_notation}
\text{PPR}(a|q_i=1-\rho, r=\rho)
\end{align}
to denote the PPR score of \attr $ a $ with respect to $ q_i $ and $ r $,
in which $ q_i $ and $ r $ are given the probability mass of $ 1-\rho $ and $ \rho $ in the personalization vector, respectively,
and $ \rho $ controls the relative importance of $ r $ in comparison to $ q_i $.

Note that PPR scores are computed for all nodes in $ \prodG $, while we are concerned about the \attrs of the recommended item.
So we normalize PPR scores such that 
the scores of recommended item's \attrs sum to one,
denoting the normalized score by nPPR.
Then, $ \p(a|q_i,r) $ is computed as
\begin{align}\label{eq:ppr_attr_relevance}
\p(a|q_i,r) \triangleq \text{nPPR}(a|q_i=1-\rho, r=\rho).
\end{align}

Similarly, we model the feedback relevance, $ \p(q_i|r) $, using PPR with respect to $ r $ over $ \prodG $,
this time normalizing PPR scores over products in $ \setFeedback $.
Thus, $ \p(q_i|r) $ is compute as
\begin{align}\label{eq:ppr_feedback_relevance}
\p(q_i|r) \triangleq \text{nPPR}(q_i|r=1).
\end{align}

By \eqref{eq:relevancescore_as_pr}, and our choice given by \eqref{eq:ppr_attr_relevance} and \eqref{eq:ppr_feedback_relevance} 
to use nPPR on $ \prodG $ to model the relevance terms,
we define $ R_r(a) $ as follows:
\begin{align}\label{eq:entityrelevance}
R_r(a)\!\triangleq\!\sum_{i=1}^{|\setFeedback_u|} {\text{nPPR}(q_i|r\!=\!1)}\!\cdot\!{\text{nPPR}(a|q_i\!=\!1\!-\!\rho, r\!=\!\rho)}.
\end{align}

Finally, based on the above definition of \attr relevance \eqref{eq:entityrelevance}, 
we define the relevance score $ R_r(\setAttribute') $ of a set $ \setAttribute' \subseteq \setAttribute $ of \attrs
with respect to recommended product $ r $ to be:
\begin{align}\label{eq:entitiesrelevance}
R_r(\setAttribute') \triangleq \sum_{a \in \setAttribute'} R_r(a).
\end{align}

We show in \Cref{sec:axiomeval,sec:realeval} that our proposed approach yields more accurate relevance scores, 
which agree with our intuition, than other choices to model the \attr relevance.

\subsubsection{\textbf{Diversity Score}}
To provide informative and engaging \es to the user, 
we want the \es to consist of diverse product \attrs.
For example, users would find it more interesting to see 
a combination of relevant reviews, product features, and purchasing history than seeing only reviews.
Diversity has multiple aspects to it, and some aspects may be application dependent.
Below we introduce two diversity aspects.
Note that our method can easily be extended to incorporate different aspects of \e diversity.

We first consider the diversity in terms of \attr types. 
Let $ a_{t} $ denote the type of \attr $ a $.
We capture this diversity using the number of \attr types covered by the selected \attrs~$ \setAttribute' $:
\begin{align}
C_{\text{Type}}(\setAttribute') = \lvert \{ a_t : a \in \setAttribute' \} \rvert.
\end{align}

Secondly, for textual product data such as customer reviews, we may consider the diversity of their topics and sentiments,
which can be extracted by existing methods, e.g., by applying the latent Dirichlet allocation to product reviews with a decision threshold.
Let $ \setTopic $ denote the set of such topics of the textual data, and 
$ \setTopic(a) \subseteq \setTopic $ be the set of topics \attr $ a $ represents.
We capture this second diversity aspect using the number of topics \attrs $ \setAttribute' $ represent, defined by:
\begin{align}
C_{\text{Topic}}(\setAttribute') = \left\lvert \bigcup\nolimits_{a \in \setAttribute'} \setTopic(a) \right\rvert.
\end{align}

\subsubsection{\textbf{Justification Score}}
Our goal is multi-objective as we aim to produce \es with high relevance and diversity.
We cast this into a single-objective optimization problem using a weighted sum scalarization.
Since relevance and diversity terms can have different magnitude, 
we adopt the normalization scheme given in~\cite{grodzevich2006normalization}
such that each term is to be bounded between 0 and 1.
For example, we define the first diversity term $ D_{\text{Type}} $ to be:
\begin{align}\label{eq:diversity:modality}
D_{\text{Type}}(\setAttribute') = \frac{ C_{\text{Type}}(\setAttribute') - C_{\text{Type}}^{\text{Min}} }{C_{\text{Type}}^{\text{Max}} - C_{\text{Type}}^{\text{Min}}},
\end{align}
where $ C_{\text{Type}}^{\text{Max}}\!=\!\max\limits_{\setAttribute_B} \{ C_{\text{Type}}(\setAttribute_B) | \setAttribute_B\!\subseteq\!\setAttribute, 0 < \lvert \setAttribute_B \rvert\!\le\!B \} $ and $ C_{\text{Type}}^{\text{Min}}\!=\!\min\limits_{\setAttribute_B} \{ C_{\text{Type}}(\setAttribute_B) | \setAttribute_B\!\subseteq\!\setAttribute, 0 < \lvert \setAttribute_B \rvert \le B \} $, with $ B $ denoting the maximum number of \attrs to be selected.
Note that $ 0 \le D_{\text{Type}}(\setAttribute')  \le 1 $.
In the event that $ C_{\text{Type}}^{\text{Max}} $ and $ C_{\text{Type}}^{\text{Min}} $ are equivalent, 
we define $ D_{\text{Type}}(\setAttribute') = 1 $.

Applying the same normalization, we define the normalized relevance score and the topical diversity term as follows:
\begin{align}
nR_r(\setAttribute') &= \left(R_r(\setAttribute') - R_r^{\text{Min}} \right) / \left(R_r^{\text{Max}} - R_r^{\text{Min}} \right)\\
D_{\text{Topic}}(\setAttribute') &= \left(C_{\text{Topic}}(\setAttribute') - C_{\text{Topic}}^{\text{Min}}\right) / \left(C_{\text{Topic}}^{\text{Max}} - C_{\text{Topic}}^{\text{Min}}\right)
\end{align}
where $ R_r^{\text{Max}}, R_r^{\text{Min}}, C_{\text{Topic}}^{\text{Max}}, C_{\text{Topic}}^{\text{Min}} $ 
are defined in the same way as in the first term, using $ R_r(\setAttribute') $ and $ C_{\text{Topic}}(\setAttribute') $ instead of $ C_{\text{Type}}(\setAttribute') $.
In sum, given a recommended item $ r $,
we define the justification score $ J_r(\setAttribute') $ of a set $ \setAttribute' \subseteq \setAttribute $ of \attrs as:
\begin{align}\label{eq:justification_score}
J_r(\setAttribute') = nR_r(\setAttribute') + \lambda_1\!\cdot\!D_{\text{Type}}(\setAttribute') + \lambda_2\!\cdot\!D_{\text{Topic}}(\setAttribute')
\end{align}
where $ \lambda_1 $ and $ \lambda_2 $ are non-negative weights for diversity terms.

\subsection{\E Discovery}

Based on the above definition of relevance, diversity, and \e scores, we formally define the justification discovery problem as follows.

\begin{tcolorbox}[boxsep=0pt,boxrule=0pt,left=0.1cm,right=0.1cm,top=0.1cm,bottom=0.1cm,enlarge top by=0.15cm]
\begin{problem}[Justification Discovery]\label{problem:objective}
Given a product graph $ \prodG $, a recommended item $ r \in \setProd $, 
products $ \setFeedback \subseteq \setProd$ that received positive feedback, and a budget $ B $,
find a set $ \setAttribute^* \subseteq \setAttribute $ of product \attrs that maximizes the \e score, i.e.,
\begin{align}\label{eq:objective}
\setAttribute^* = \argmax_{\setAttribute' \subseteq \setAttribute} J_r(\setAttribute') \text{ such that } |\setAttribute'| \le B.
\end{align}
\end{problem}
\end{tcolorbox}
Due to the combinatorial nature of this optimization problem, 
solving it exactly is computationally intractable.
Instead, we show that the objective \eqref{eq:objective} is submodular, 
which allows us to efficiently obtain near-optimal \es.

\begin{theorem}\label{thm:submodular}
The justification score $ J_r(\setAttribute') $ given by \eqref{eq:justification_score} is a non-negative, monotone, submodular function.
\end{theorem}
\begin{proof}
(a) $ J_r(\setAttribute') $ is non-negative since it is a weighted sum of three non-negative scores with non-negative weights.

(b) A set function $ f $ is monotone if for every $ \setAttribute_1 \subseteq \setAttribute_2 $, we have that $ f(\setAttribute_1) \le f(\setAttribute_2) $.
Given that
{\small\begin{align*}
\small
R_r(\setAttribute_1) =& \sum\nolimits_{a \in \setAttribute_1} R_r(a) \le \sum\nolimits_{a \in \setAttribute_2} R_r(a) = R_r(\setAttribute_2),\\
C_{\text{Type}}(\setAttribute_1) =& \lvert \{ a_t \!:\! a \in \setAttribute_1 \} \rvert \le \lvert \{ a_t\!:\!a \in \setAttribute_2 \} \rvert = C_{\text{Type}}(\setAttribute_2),\\
C_{\text{Topic}}(\setAttribute_1) =& \left\lvert \bigcup\nolimits_{a \in \setAttribute_1} \setTopic(a) \right\rvert \le \left\lvert \bigcup\nolimits_{a \in \setAttribute_2} \setTopic(a) \right\rvert = C_{\text{Topic}}(\setAttribute_2),
\end{align*}}$ J_r(\setAttribute') $ is monotone as it is a non-negative weighted sum of these monotone scores. 

(c) Given a set function $ f $ of \attrs $ \setAttribute' $, and \attr $ a $,
let $ \Delta f(a|\setAttribute')\!=\!f(\setAttribute' \cup \{a\}) - f(\setAttribute') $ be the marginal gain of adding $ a $ to $ \setAttribute' $.
Then $ f $ is submodular if for every $ \setAttribute_1, \setAttribute_2 \subseteq \setAttribute $ with $ \setAttribute_1\!\subseteq\!\setAttribute_2 $ and 
every $ a \in \setAttribute\!\setminus\!\setAttribute_2 $, 
it holds that $ \Delta f(a|\setAttribute_1) \ge \Delta f(a|\setAttribute_2) $.

As every $ a\!\in\!\setAttribute \setminus \setAttribute_2 $ results in the same marginal gain $ R_r(a) $ to both $ R_r(\setAttribute_1) $ and $ R_r(\setAttribute_2) $,
$ \Delta R_r(a | \setAttribute_1) = \Delta R_r(a | \setAttribute_2) $.

Let $ \setTopic(\setAttribute') = \bigcup\nolimits_{a \in \setAttribute'} \setTopic(a) $. 
The topics $ \setTopic(a) $ covered by $ a $ can be classified into three cases.
The first case is the set $ \setTopic_{1} \subseteq \setTopic(a) $ of topics that belong to $ \setTopic(\setAttribute_1) $. 
Since $ \setAttribute_1 \subseteq \setAttribute_2 $, $ \setTopic_1 \subseteq \setTopic(\setAttribute_2) $; 
thus, these topics make no additional contributions to both $ D_{\text{Topic}}(\setAttribute_1) $ and $ D_{\text{Topic}}(\setAttribute_2) $.
The second case is the set $ \setTopic_{2} \subseteq \setTopic(a) $ of topics 
that do not belong to $ \setTopic(\setAttribute_1) $, but belong to $ \setTopic(\setAttribute_2) $.
Since these topics already belong to $ \setTopic(\setAttribute_2) $,
they make positive contributions to $ D_{\text{Topic}}(\setAttribute_1) $, while making no contributions to $ D_{\text{Topic}}(\setAttribute_2) $.
The third case is the set $ \setTopic_{3} \subseteq \setTopic(a) $ of topics 
that do not belong to both $ \setTopic(\setAttribute_1) $ and $ \setTopic(\setAttribute_2) $.
In this case, the topics in $ \setTopic_{3} $ make equal contributions to $ D_{\text{Topic}}(\setAttribute_1) $ and $ D_{\text{Topic}}(\setAttribute_2) $.
Thus, in all cases, $ \Delta D_{\text{Topic}}(a | \setAttribute_1) \ge \Delta D_{\text{Topic}}(a | \setAttribute_2) $.
The same argument applies to show that $ \Delta D_{\text{Type}}(a | \setAttribute_1) \ge \Delta D_{\text{Type}}(a | \setAttribute_2) $.

Therefore, since three functions $ R_r(\setAttribute'), D_{\text{Type}}(\setAttribute') $, and $ D_{\text{Topic}}(\setAttribute') $ are all submodular, 
the justification score $ J_r(\setAttribute') $, which is a weighted sum of these submodular functions with non-negative weights, is also submodular.
\end{proof}

\begin{theorem}\label{thm:approximation}
\Cref{problem:objective} admits a $ (1\!-\!\frac{1}{e}) $-approximation.
\end{theorem}
\begin{proof}
Maximizing a non-negative, monotone, submodular function subject to a cardinality constraint 
admits a $ (1\!-\!\frac{1}{e}) $ approximation
under a greedy approach in which the item with the largest marginal gain is selected at each step~\cite{DBLP:journals/mor/NemhauserW78}.
This theorem follows since our maximization objective $ J_r(\setAttribute') $ is non-negative, monotone, and submodular by \Cref{thm:submodular}.
\end{proof}

\textbf{\method algorithm.} \Cref{thm:approximation} leads to the \method algorithm in~\Cref{alg:framework},
which finds a $ (1\!-\!\frac{1}{e}) $-approximation of the optimal \es.
In \Cref{alg:framework}, we first compute the relevance score of all \attrs based on \eqref{eq:entityrelevance}.
Then, we repeatedly find the product \attr from $ \setAttribute_{(r)} $ with the greatest marginal gain, 
where $ \setAttribute_{(r)} $ is the \attrs of product $ r $, and
add it to $ \setAttribute^* $ until we exhaust the given budget $ B $.

\setlength{\textfloatsep}{0pt}
\begin{algorithm}[!t]
\DontPrintSemicolon
\SetNoFillComment
\SetKwComment{Comment}{$\triangleright$\ }{}
\KwIn{Product graph $ \prodG $, recommended product $ r $, products $ \setFeedback $ with positive feedback, budget $ B $.}
\KwOut{\ES $ \setAttribute^* $ s.t. $ |\setAttribute^*| \le B $.}

Compute the relevance score $ R_r(a) $ given by \eqref{eq:entityrelevance}.\label{alg:framework:relscore}\\
$ \setAttribute^* \gets \emptyset $\\
\While{$ |\setAttribute^*| \le B $}{
	$ a^* \gets \argmax_{a \in \setAttribute_{(r)} \setminus \setAttribute^*} \left( J_r(\setAttribute^* \cup \{a\}) - J_r(\setAttribute^*) \right) $\\
	$ \setAttribute^* \gets \setAttribute^* \cup \{ a^* \} $\\
}
\Return $ \setAttribute^* $
\caption{\method algorithm.}
\label{alg:framework}
\end{algorithm}

\begin{theorem}
\Cref{alg:framework} runs in $ O(|E|) $ time with $ O(|\setFeedback|) $ processors,
taking $ O(|E||\setFeedback|) $ steps in total, assuming that $ B M < |V| $,
where $ M = \max( N_1, N_2 ) $ with $ N_1 $ and $ N_2 $ denoting the number of \attr types and topics, respectively.
\end{theorem}
\begin{proof}
Computing the relevance score in \eqref{eq:entityrelevance} for all attributes
involves computing the \ppr scores with respect to $ |\setFeedback| + 1 $ personalization vectors.
Using a power iteration with sparse matrix multiplications, \ppr can be computed in $ O(|E|) $ time.
Since \ppr computations are independent of each other, they can be completed in $ O(|E|) $ time using $ O(|\setFeedback|) $ processors.

Computing $ C_{\text{Topic}}^{\text{Max}} $ corresponds to the maximum coverage problem.
As this is NP-hard, we use a greedy approximation algorithm~\cite{kleinberg2006algorithm}, 
which takes $ O(BM) $ time. Other terms for score normalization (e.g., $ C_{\text{Type}}^{\text{Max}} $) can also be computed in $ O(BM) $ time.
Then we select up to $ B $ \es: 
Selecting one \e requires evaluating the marginal gain $ J_r(\setAttribute^* \cup \{a\}) - J_r(\setAttribute^*) $ 
for all $ a \in \setAttribute_{(r)} \setminus \setAttribute^*$.
As evaluating the marginal gain with respect to each term in \eqref{eq:justification_score} takes $ O(M) $,
greedy selection takes $ O(B M |\setAttribute_{(r)}|) $ steps in total.
Given that $ |\setAttribute_{(r)}| < |V| $ and assuming $ BM < |V| $, which is true in most cases, 
the running time for the greedy selection is $ O(|E|) $.
\end{proof}

%% file: 030axiomeval.tex
In this section, we evaluate different approaches to measure \attr relevance, 
using what we call \textit{axioms}, which are a set of \pgs with an intuitive expected outcome.
After describing axioms, 
we introduce other approaches to compute \attr relevance, and discuss how well axioms are satisfied by different approaches.
Experimental settings used for this evaluation are given in~\Cref{sec:appendix:settings}.

\begin{figure*}[!htbp]
\par\vspace{-0.5em}\par
\centering
\setlength{\abovecaptionskip}{1pt}
\makebox[\linewidth][c]{\includegraphics[width=1.02\linewidth]{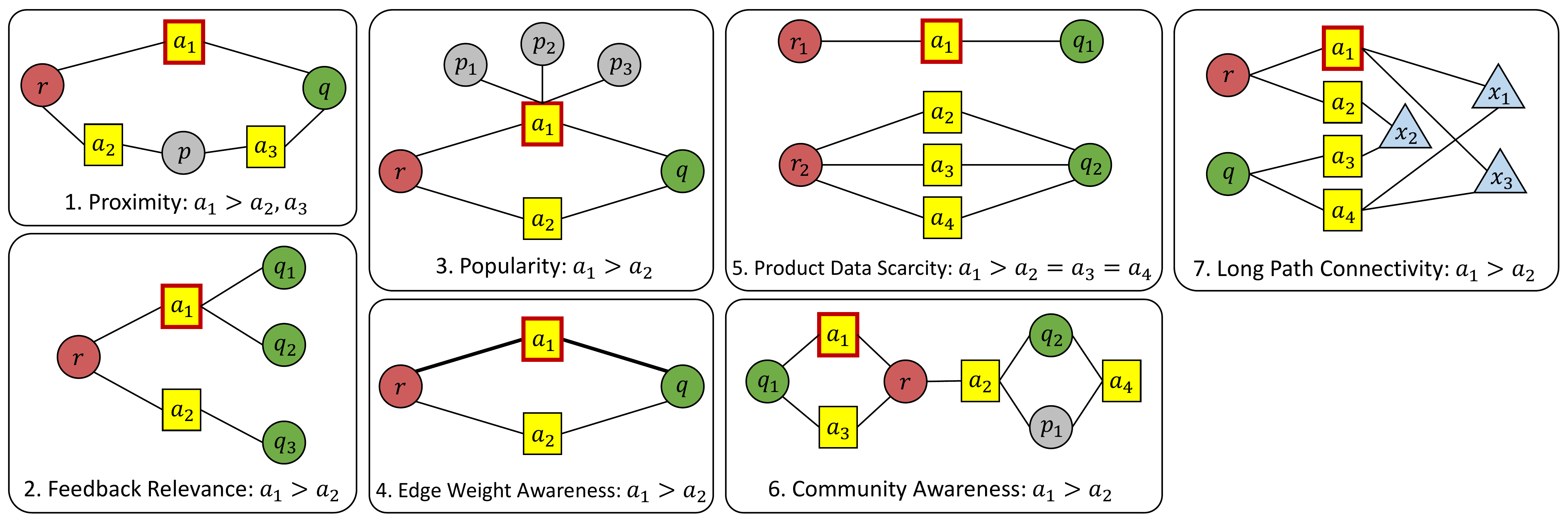}}
\caption{Axioms and expected relevance scores \eqref{eq:entityrelevance} of product attributes (see \Cref{sec:axiomeval:axioms} for details).}
\label{fig:axioms}
\end{figure*}

\begin{table*}[!t]
\setlength{\tabcolsep}{1.5mm}
\setlength{\abovecaptionskip}{4pt}
\tablefont
\centering
\caption{Our proposed \method satisfies all axioms in~\Cref{fig:axioms}, while other alternatives fail at least one of them.}
\makebox[0.4\textwidth][c]{\setlength\doublerulesep{0.5pt}
	\begin{tabular}{c | c | c | c | c | c | c | c | c}
		\toprule
		\multicolumn{2}{c|}{\diagbox{\textbf{Method}}{\textbf{Axioms}}} & \textbf{\makecell{1. Proximity}} & \textbf{\makecell{2. Feedback\\Relevance}} & \textbf{\makecell{3. Popularity\\}} & \textbf{\makecell{4. Edge Weight\\Awareness\\}} & \textbf{\makecell{5. Product Data\\Scarcity\\}} & \textbf{\makecell{6. Community\\Awareness\\}} & \textbf{\makecell{7. Long Path\\Connectivity\\}} \\
		\midrule \midrule
		\belowrulesepcolor{Gray}
		\rowcolor{Gray}
		\multicolumn{2}{c|}{\textbf{\method}} & $\boldcheckmark$ & $\boldcheckmark$ & $\boldcheckmark$ & $\boldcheckmark$ & $\boldcheckmark$ & $\boldcheckmark$ & $\boldcheckmark$ \\ \aboverulesepcolor{Gray}\midrule \midrule 
		\multicolumn{2}{c|}{ExpLOD~\cite{DBLP:conf/recsys/MustoNLGS16}} & $ \checkmark $ & & & & & $ \checkmark $ & \\
	 	\multicolumn{2}{c|}{\mpand~\cite{DBLP:conf/kdd/TongF06}} & $ \checkmark $ & $ \checkmark $ & $ \checkmark $ & $ \checkmark $ & $ \checkmark $ & & $ \checkmark $ \\
	 	\multicolumn{2}{c|}{\mpor~\cite{DBLP:conf/kdd/TongF06}} & $ \checkmark $ & $ \checkmark $ & $ \checkmark $ & $ \checkmark $ & $ \checkmark $ & & $ \checkmark $ \\
		\multicolumn{2}{c|}{ $ \mathrm{BA}_{q \rightarrow r}^{\text{Sink}} $~\cite{DBLP:conf/pakdd/TongPFYE10} } & $ \checkmark $ & $ \checkmark $ &  & $ \checkmark $ & $ \checkmark $ & $ \checkmark $ & \\ 
		\multicolumn{2}{c|}{ $ \mathrm{BA}_{q \rightarrow r}^{\text{Del}} $~\cite{DBLP:conf/pakdd/TongPFYE10} } & $ \checkmark $ & $ \checkmark $ &  &  & $ \checkmark $ &  & \\ 
		\multicolumn{2}{c|}{ $ \mathrm{BA}_{r \rightarrow q}^{\text{Sink}} $~\cite{DBLP:conf/pakdd/TongPFYE10} } & $ \checkmark $ & $ \checkmark $ &  & $ \checkmark $ & $ \checkmark $ & $ \checkmark $ & \\ 
		\multicolumn{2}{c|}{ $ \mathrm{BA}_{r \rightarrow q}^{\text{Del}} $~\cite{DBLP:conf/pakdd/TongPFYE10} } & $ \checkmark $ & $ \checkmark $ &  &  & $ \checkmark $ & $ \checkmark $ & \\ \midrule 
		\mpand &	\multirow{6}{*}{\rotatebox[origin=c]{90}{\makecell{w/ HAR~\cite{DBLP:conf/sdm/LiNY12}\\authority score}}} & $ \checkmark $ & $ \checkmark $ & $ \checkmark $ & $ \checkmark $ & $ \checkmark $ &  & \\
		\mpor &	 &  & $ \checkmark $ & $ \checkmark $ & $ \checkmark $ & $ \checkmark $ &  & \\
		$ \mathrm{BA}_{q \rightarrow r}^{\text{Sink}} $ & 	& $ \checkmark $ & $ \checkmark $ &  &  & $ \checkmark $ & & \\ 
		$ \mathrm{BA}_{q \rightarrow r}^{\text{Del}} $ & 	& $ \checkmark $ &  &  &  & $ \checkmark $ &  & \\ 
		$ \mathrm{BA}_{r \rightarrow q}^{\text{Sink}} $ &	& $ \checkmark $ & $ \checkmark $ &  & $ \checkmark $ & $ \checkmark $ & $ \checkmark $ &  \\ 
		$ \mathrm{BA}_{r \rightarrow q}^{\text{Del}} $ &  &	&  &  &  & $ \checkmark $ & $ \checkmark $ & $ \checkmark $ \\ 
		\bottomrule
	\end{tabular}
}
\vspace{-1.0em}
\label{tab:result:axioms}
\end{table*}

\subsection{Axioms}
\label{sec:axiomeval:axioms}

An axiom is a small product graph with expected relevance scores \eqref{eq:entityrelevance} for some product \attrs.
We use the axioms shown in \Cref{fig:axioms} to evaluate different methods for measuring \attr relevance.
In \Cref{fig:axioms}, squares denote product \attrs, and 
circles denote products, where $ r $ is the recommended product, $ q $ is the product with positive user feedback, and 
$ p $ is a product the user had no interaction with.

\textbf{1. Proximity.} 
\Attrs that are closer to products $ r $ and $ q $ should receive a higher relevance score.
Given that $ a_1 $ is directly connected to $ r $ and $ q $ (i.e., $ a_1 $ is a shared attribute of $ r $ and $ q $), 
while other \attrs are three hops away from either $ r $ or $ q $, 
we expect $ R_r(a_1) > R_r(a_2), R_r(a_3) $.

\textbf{2. Feedback Relevance.}
\Attrs that receive more support from the products in $ \setFeedback $ should receive a higher relevance score.
While both $ a_1 $ and $ a_2 $ are directly connected to $ r $, $ a_1 $ is covered by a greater number of products in $ \setFeedback $.
Thus, we expect $ R_r(a_1) > R_r(a_2) $.

\textbf{3. Popularity.}
More widely used \attrs should receive a higher relevance score
(e.g., consider explaining a movie recommendation with a popular actor vs. an actor who appeared only in that movie).
Thus, we require that $ R_r(a_1) > R_r(a_2) $.

\textbf{4. Edge Weight Awareness.}
\Attrs that are connected via edges of higher weight should receive a higher relevance score, 
as edge weight indicates the importance of each connection in the product graph.
Thus, we expect $ R_r(a_1) > R_r(a_2) $.
Note that satisfying this axiom is important to reflect the relative importance among different \attr types.
For instance, while country \attrs would be one of the most popular ones in the movie product graph,
it may not be a very interesting \e to the user.
Based on this prior knowledge, we can downplay the country type if our method satisfies this axiom.

\textbf{5. Product Data Scarcity.}
\Attrs of the product that contain scarcer information should receive a higher relevance score.
This axiom consists of two product graphs, in which each \attr is directly connected to $ r $ and $ q $.
While each graph contains only one product in $ \setFeedback $, 
the positive feedback expressed by $ \setFeedback $ is attributed to multiple \attrs in the larger graph, 
leading to each one of them being a weaker evidence of user's preference than the other product's unique \attr.
Thus, we expect $ R_r(a_1) > R_r(a_i) $ for $ i \ne 1 $.

\textbf{6. Community Awareness.}
When the recommended item belongs to one community (e.g., a group of shoes),
it is desirable to put more weight on \attrs belonging to the same community than on others (e.g., a group of electronic devices),
as similar products tend to share more \attrs (e.g., product details, keywords) with each other than with different products.
In this axiom, there are two small communities of products and \attrs,
where $ r $ and $ a_2 $ act as a bridge between communities.
Also, each community has only one product with positive feedback.
Therefore, we expect $ R_r(a_1) > R_r(a_2) $.

\textbf{7. Long Path Connectivity.}
Even when \attrs are not directly connected to $ r $ and $ q $, 
\attrs more strongly connected to $ r $ and $ q $ should receive a higher relevance score.
Here, $ a_1 $ is more strongly linked to $ q $ via $ x_1 $ and $ x_3 $ than $ a_2 $, 
which is linked to $ q $ only via $ x_2 $.
Thus, we expect $ R_r(a_1) > R_r(a_2) $.

\subsection{Baselines}\label{sec:axiomeval:baseline}

Below we use the same notation as in \eqref{eq:ppr_notation} to specify the query nodes. \Cref{tab:symbols} provides the definition of symbols.

\subsubsection{Relevance Models} We consider three approaches that measure the relevance of \attr $ a $, 
given recommendation $ r $ and products $ q_1,\ldots,q_{|\setFeedback_u|} $.

\textbf{ExpLOD}~\cite{DBLP:conf/recsys/MustoNLGS16} assigns a high score to those \attrs 
that are highly connected to the products in $ \setFeedback_u \cup \{r\} $, using the following formula:
\begin{align}
\text{ExpLOD}(a) = \left( \alpha \cdot \frac{n_{a,\setFeedback_u}}{|\setFeedback_u|} + \beta \cdot n_{a,r} \right) \cdot \mathrm{IDF}_{a}
\end{align}
where $ n_{a,\setFeedback_u} $ is the number of edges between product \attr $ a $ and the products in $ \setFeedback_u $,
$ n_{a,r} $ is the number of edges between $ a $ and $ r $, and
$ \mathrm{IDF}_{a} $ is the reciprocal of the number of products that are described by \attr $ a $.

\textbf{Meeting Probability (MP)}~\cite{DBLP:conf/kdd/TongF06}
assigns a high relevance score to an \attr that is close to the recommended item~$ r $ and products $ |\setFeedback_u| $.
MP has been successfully used in identifying nodes that have strong connections to the query nodes.
We consider the following two MP scores.
\begin{align}
\text{\mpand}(a) &\triangleq \prod_{p \in \setFeedback_u \cup \{r\}} {\mathrm{PPR}(a|p\!=\!1)} \label{eq:mp} \\
\text{\mpor}(a) &\triangleq 1 -\!\prod_{p \in \setFeedback_u \cup \{r\}} \left(1 - \mathrm{PPR}(a|p\!=\!1)\right) \label{eq:mpor}
\end{align}

\textbf{BASSET (BA)}~\cite{DBLP:conf/pakdd/TongPFYE10} 
aims to identify a small number of good gateway nodes between a source node $ s $ and a target node $ t $ in the given graph.
Let $ \mathrm{PPR}_{\mathcal{I}}^{\text{Sink}}(t|s=1) $ denote the PPR score from source $ s $ to target $ t $,
after setting the nodes denoted by $ \mathcal{I} $ as sinks (i.e., nodes with no outgoing edges).
We also consider a related model, $ \mathrm{PPR}_{\mathcal{I}}^{\text{Del}}(t|s=1) $, 
in which we delete the nodes denoted by~$ \mathcal{I} $, instead of making them sinks.
Note that, in out setting, given a recommendation $ r $ and a product $ q $, we can consider two directions of $ q \rightarrow r $ and $ r \rightarrow q $.
For instance, BA score of \attr $ a $ between source $ q $ and target $ r $ using $ \mathrm{PPR}_{\mathcal{I}}^{\text{Sink}} $ is defined as:
\begin{align}
\mathrm{BA}_{q \rightarrow r}^{\text{Sink}}(a) \triangleq \sum_{i=1}^{|\setFeedback_u|} {\mathrm{PPR}(r|q_i\!=\!1)} - {\mathrm{PPR}_{\{a\}}^{\text{Sink}}(r|q_i\!=\!1)}.
\end{align}
Three other options, $ \mathrm{BA}_{q \rightarrow r}^{\text{Del}}(a) $, $ \mathrm{BA}_{r \rightarrow q}^{\text{Sink}}(a) $, 
and $ \mathrm{BA}_{r \rightarrow q}^{\text{Del}}(a) $ are defined analogously.

\subsubsection{Proximity Measure}

A proximity measure provides a way to compute node proximity with respect to a query node.
We consider two important proximity measures.

\textbf{PPR} (Personalized PageRank)~\cite{DBLP:conf/www/Haveliwala02} measures node-to-node proximity 
by the limiting probability distribution of a random walker biased towards a set of query nodes.

\textbf{HAR}~\cite{DBLP:conf/sdm/LiNY12} is a generalization of SALSA~\cite{DBLP:journals/tois/LempelM01} to handle multi-relational data, 
which enables users to specify the relative importance of relations.
HAR has been shown to outperform SALSA~\cite{DBLP:journals/tois/LempelM01} and HITS~\cite{DBLP:journals/jacm/Kleinberg99} 
in identifying relevant results to the query input.
HAR computes hub score and authority score with respect to a query entity.
While we considered both scores as a proximity measure,
due to space constraints, we report only the result obtained with authority score as hub score was mostly outperformed by authority score.

Among the relevance models introduced above, MP and BA internally use PPR.
We evaluate variants of these baselines using HAR as their proximity measure (except for ExpLOD, which is not a random walk-based method).

\subsection{Results}

\Cref{tab:result:axioms} summarizes which axioms are satisfied by different approaches to measure \attr relevance.
A checkmark indicates that the expected outcome of the axiom has been achieved by the corresponding method.
\method is the only one that satisfies all axioms.
Other alternatives fail at least one of the axioms;
among them, \mpr is the next best one, failing only (6) Community Awareness axiom.
In~\Cref{sec:realeval}, we show that \method also leads to more accurate \es than \mpr in experiments using real-world product graphs.

\explod does not satisfy most of the axioms, mainly due to the fact that 
it can only consider direct edges between products and \attrs, failing to propagate information over the graph.

In general, BASSET (BA) led to worse results than \method and \mpr, 
failing (1) Popularity, (6) Community Awareness, and (7) Long Path Connectivity axioms in many cases,
indicating that good gateway nodes may not serve well as a \e.
Also, while \har achieved a reasonably good result (for instance, with \mpand), 
relevance models could satisfy more axioms using \ppr as a proximity measure.

%% file: 040realeval.tex
In this section, we address the following questions.

\begin{enumerate}[label=\textbf{Q{{\arabic*}}.},ref=Q\arabic*]
\item \label{sec:realeval:q1} \textbf{\E Quality.} 
How well does \method justify recommendations?

\item \label{sec:realeval:q2} \textbf{Scalability.}
How does \method scale up with the increase of the input size?

\item \label{sec:realeval:q3} \textbf{Relevance-Diversity Trade-Off.}
How does increasing the weight for diversity affect the relevance of \es?
\end{enumerate}
Experimental settings are given in~\Cref{sec:appendix:settings}.

\begin{table}[!t]
\centering
\setlength{\tabcolsep}{1.0mm}
\setlength{\abovecaptionskip}{3pt}
\caption{
	Statistics of real-world product graphs.
}
\makebox[0.4\textwidth][c]{
\begin{tabular}{ c | c | r | r | r }
	\toprule
	\textbf{Name} & \textbf{\makecell{Product}} & \textbf{\makecell{\# Products}} & \textbf{\makecell{\# Nodes}} & \textbf{\makecell{\# Edges}} \\ 
	\midrule
	\movie & Movie & 4,803 & 308,304 & 1,329,428 \\
	\paper & Paper & 11,941 & 111,007 & 724,962 \\
	\paperlarge & Paper & 2,094,396 & 7,426,773 & 100,000,000 \\
	\bottomrule
\end{tabular}
}
\label{tab:datasets}
\end{table}

\vspace{-0.5em}
\subsection{Datasets}
\vspace{-0.5em}

We construct product graphs from public datasets on movies and publications.
\Cref{tab:datasets} shows the statistics of our datasets.

\textbf{\movie} consists of movies and movie-related attributes, such as actors, directors, crews, casts, and movie keywords, 
extracted from the TMDb 5000 movie dataset\footnote{https://www.kaggle.com/tmdb/tmdb-movie-metadata}.
We also added positive reviews to \movie, which gave 10/10 rating to the movies.
Reviews were retrieved from the IMDb website\footnote{https://www.imdb.com/}, instead of TMDb, 
since more reviews were available on IMDb.
We extracted keywords from each review, by filtering out those words whose tf-idf score is below a threshold, 
and connected them with the reviews.

\textbf{\paper} is a product graph of papers and related attributes, such as authors, citations, publication venues, and fields of study,
constructed from the citation network dataset v12~\cite{DBLP:conf/kdd/TangZYLZS08}\footnote{https://www.aminer.org/citation}.
In \paper, we included papers published at KDD, SIGMOD, and ICML, and their attributes,
while excluding papers that were cited less than 5 times.
We also created \paperlarge that contains $ 10^8 $ edges for scalability evaluation,
which is also constructed from the same citation network dataset, and includes all venues and papers.

\subsection{Baselines}
Among the baselines used in \Cref{sec:axiomeval},
we use \mpand and \mpor~\cite{DBLP:conf/kdd/TongF06}, which satisfied most axioms among baselines, and
\explod~\cite{DBLP:conf/recsys/MustoNLGS16}, which is a representative non-random walk based method for justifying recommendations.
Among them, we exclude BA and \har (which was used as a proximity measure)
as they were mostly outperformed by other alternatives.
We also include \pr~\cite{page1999pagerank}, which estimates \attr relevance by its PageRank score in the product graph.

\subsection{\ref{sec:realeval:q1}. \E Quality}\label{sec:realeval:quality}

We evaluate the quality of \es in two ways: automatic evaluation and qualitative analysis.

\subsubsection{{Automatic Evaluation}}
For an automatic and objective evaluation, 
we consider the task of user preference retrieval.

\textbf{User Preference Retrieval.}
Among several attributes of the recommended product $ r $, 
we want those that reflect the user's preference better than others
to receive higher relevance scores and be used as a \e.
In our two datasets, the review and the paper written by a user clearly reflect the user's preference.
Thus, among the reviews and papers associated with recommended product $ r $,
it is desirable for the review and paper written by the user to get higher scores than others.

For \movie, we select a positive review written by the user who wrote at least 10 positive reviews.
Note that the movies for which user $ u $ wrote positive reviews correspond to $ Q_u $.
Similarly, for \paper, we select a reference written by the user who published at least 15 papers.
In total, we randomly select \numSelection reviews and \numSelection citations.

\textbf{Performance Evaluation.} We evaluate preference retrieval results using the mean reciprocal rank (MRR).
Let $ \setAttribute_R $ denote the chosen attributes of a specific type (e.g., reviews or cited papers), written by different users.
Let $ p_i $ refer to the product of $ i $-th attribute in $ \setAttribute_R $ (e.g., the product for which the $ i $-th review was written), 
and $ \text{rank}_i $ be the rank position of the $ i $-th attribute among the corresponding attributes of product $ p_i $,
where ranks are determined by the estimated relevance score computed with respect to $ r $ and 
the products $ Q_u $ that received positive feedback by user $ u $ who created the $ i $-th attribute 
(e.g., movies that received positive reviews by user $ u $).
MRR is defined by
\begin{align}
\setlength{\abovedisplayskip}{1pt}
\setlength{\belowdisplayskip}{\abovedisplayskip}
\setlength{\abovedisplayshortskip}{0pt}
\setlength{\belowdisplayshortskip}{1pt}
\text{MRR} = \frac{1}{|\setAttribute_R|} \sum\limits_{i=1}^{|\setAttribute_R|} \frac{1}{\text{rank}_i},
\end{align}
and higher values are better.

\textbf{Results.}
\Cref{fig:exp:crown:prefretrieval} shows the MMR on two datasets.
\method achieved the best MMR, which is up to 20.7\% higher than the second best result achieved by \mpand.
While \mpor outperformed \mpand on \paper by a small margin, 
results show that \mpor's performance is more sensitive to the dataset than \mpand.
\pr's performance is worse than \method as it does not consider $ Q_u $ and $ r $ in measuring attribute relevance.
Since each review applies to only one product, 
\explod ends up assigning identical scores to all reviews, making it inapplicable to be used to retrieve user preferences on \movie.
On \paper, \explod obtains even lower MMR than \pr,
which is due to the fact that \explod computes relevance score based only on the direct connection between products and attributes,
failing to propagate user preference over a graph.

Note that since we ignore the relevance of other reviews and citations, which is unknown to us, 
these results are a lower bound of the true MMR.
As MMR considers the rank of the first relevant entity, 
the true MMR would be higher than the current result 
if there is another relevant attribute, ranked higher than the user's review and paper.

\begin{table*}[!t]
\centering
\setlength{\aboverulesep}{0pt}
\setlength{\belowrulesep}{0pt}
\setlength{\extrarowheight}{1pt}
\setlength{\tabcolsep}{0.2em}
\setlength{\abovecaptionskip}{4pt}
\caption{\method produces qualitatively better results than \mpand. The table shows the top-15 references 
	cited by the paper ``Rubik: Knowledge Guided Tensor Factorization and Completion for Health Data Analytics''~\cite{DBLP:conf/kdd/WangCGDKCMS15},
	ordered by the relevance computed with \method (with $ \rho $ set to $ 0.5 $ and $ 0.9 $) and \mpand.
	Papers on electronic health record analysis are highlighted in blue font, and 
	the two highly relevant papers that employ tensor factorization for electronic health record analysis are further highlighted in cyan background color.
}
\makebox[0.4\textwidth][c]{
	\begin{tabular}{ p{0.355\textwidth} | p{0.355\textwidth} | p{0.355\textwidth} }
		\toprule
		\makecell[c]{\textbf{\method} ($ \rho=0.5 $)} & \makecell[c]{\textbf{\method} ($ \rho=0.9 $)} & \makecell[c]{\textbf{\mpand}~\cite{DBLP:conf/kdd/TongF06}} \\
		\midrule
		Positive tensor factorization & Tensor decompositions and applications & Tensor decompositions and applications \\ \midrule
		Tensor decompositions and applications & Positive tensor factorization & Distributed optimization and statistical learning via the alternating direction method of multipliers \\ \midrule
		\cyancell{\highlight{Marble: high-throughput phenotyping from electronic health records via sparse nonnegative tensor factorization}} & \cyancell{\highlight{Marble: high-throughput phenotyping from electronic health records via sparse nonnegative tensor factorization}} & Positive tensor factorization \\ \midrule
		On tensors, sparsity, and nonnegative factorizations & \cyancell{\highlight{Limestone: high-throughput candidate phenotype generation via tensor factorization}} & Scalable tensor factorizations for incomplete data \\ \midrule
		\cyancell{\highlight{Limestone: high-throughput candidate phenotype generation via tensor factorization}} & On tensors, sparsity, and nonnegative factorizations & Learning with tensors: a framework based on convex optimization and spectral regularization \\ \midrule
		Distributed optimization and statistical learning via the alternating direction method of multipliers & Scalable tensor factorizations for incomplete data & \cyancell{\highlight{Marble: high-throughput phenotyping from electronic health records via sparse nonnegative tensor factorization}} \\ \midrule
		Scalable tensor factorizations for incomplete data & Distributed optimization and statistical learning via the alternating direction method of multipliers & A block coordinate descent method for regularized multiconvex optimization with applications to nonnegative tensor factorization and completion \\ \midrule
		Tensor completion for estimating missing values in visual data & Tensor completion for estimating missing values in visual data & On tensors, sparsity, and nonnegative factorizations \\ \midrule
		\highlight{Network discovery via constrained tensor analysis of fMRI data} & \highlight{Network discovery via constrained tensor analysis of fMRI data} & Tensor completion for estimating missing values in visual data \\ \midrule
		Learning with tensors: a framework based on convex optimization and spectral regularization & \highlight{Next-generation phenotyping of electronic health records} & \highlight{Network discovery via constrained tensor analysis of fMRI data} \\ \midrule
		\highlight{Next-generation phenotyping of electronic health records} & Learning with tensors: a framework based on convex optimization and spectral regularization & \cyancell{\highlight{Limestone: high-throughput candidate phenotype generation via tensor factorization}} \\ \midrule
		Square deal: lower bounds and improved relaxations for tensor recovery & Square deal: lower bounds and improved relaxations for tensor recovery & \highlight{Next-generation phenotyping of electronic health records} \\ \midrule
		FlexiFaCT: scalable flexible factorization of coupled tensors on Hadoop & FlexiFaCT: scalable flexible factorization of coupled tensors on Hadoop & Square deal: lower bounds and improved relaxations for tensor recovery \\ \midrule
		All-at-once optimization for coupled matrix and tensor factorizations & All-at-once optimization for coupled matrix and tensor factorizations & Convex tensor decomposition via structured Schatten norm regularization \\ \midrule
		A block coordinate descent method for regularized multiconvex optimization with applications to nonnegative tensor factorization and completion & Convex tensor decomposition via structured Schatten norm regularization & A new convex relaxation for tensor completion \\
		\bottomrule
	\end{tabular}
}
\vspace{-1.5em}
\label{tab:result:qualitative1}
\end{table*}

\subsubsection{Qualitative Analysis}

We present two case studies where we compare the results obtained with \method and \mpand on \paper.
To see how the parameter $ \rho $ in \eqref{eq:entityrelevance} affects \method, 
we report two results for \method using $ \rho=0.5 $ and $ \rho=0.9 $.

\textit{Case 1.} 
\method and \mpand are given the paper entitled 
``Rubik: Knowledge Guided Tensor Factorization and Completion for Health Data Analytics''~\cite{DBLP:conf/kdd/WangCGDKCMS15} as a recommended item $ r $, 
and the set $ Q_u $ with ten papers on matrix and tensor factorizations.
\Cref{tab:result:qualitative1} shows top-15 papers cited by the recommended paper,
ordered by the relevance computed with \method~\eqref{eq:entityrelevance} and \mpand~\eqref{eq:mp}.
The papers in blue font deal with electronic health record (EHR) analysis, and
those highlighted in cyan background color are two highly relevant papers that employ tensor factorization for EHR analysis.
Among the citations, these highlighted papers are particularly relevant \es,
since they cut across two topics central to the recommended paper and the papers in $ Q_u $, 
namely, EHR analysis and tensor factorization.
In \Cref{tab:result:qualitative1}, these two highly relevant papers belong to the top-5 citations in both results of \method,
while they are ranked at 6th and 11th places in the result of \mpand.
Further, since EHR analysis is a major topic of the recommended paper,
with $ \rho=0.9 $, \method gives more weight to EHR analysis than with $ \rho=0.5 $,
which boosts the ranking of the papers on EHR analysis.

\textit{Case 2.}
\method and \mpand are given the paper entitled ``Towards Parameter-Free Data Mining''~\cite{DBLP:conf/kdd/KeoghLR04} as a recommended item $ r $, 
and the set $ Q_u $ with ten papers on time series analysis.
\Cref{tab:result:qualitative2} shows top-15 papers cited by the recommended paper,
ordered by the relevance computed with \method~\eqref{eq:entityrelevance} and \mpand~\eqref{eq:mp}.
The citations in blue font are the papers relevant to the MDL principle.
Among the cited publications, those on MDL are highly relevant \es
as MDL principle is central to the main idea of the recommended paper.
At the same time, since $ Q_u $ contains papers on time series analysis, references on time series are relevant.
In \Cref{tab:result:qualitative2}, \method retrieves four papers on MDL (ranked at 2nd, 7th, 11th, and 13th places with $ \rho=0.9 $), 
and eleven papers on time series.
On the other hand, \mpand retrieves only two papers on MDL (ranked at 7th and 14th places).
Also, in this case, increasing $ \rho $ to $ 0.9 $ led to a more drastic change than in the previous case, 
boosting the ranking of the papers on MDL.

Overall, in these case studies, \method produces qualitatively better results than \mpand,
which are more balanced in terms of the relevance to the recommendation, user preferences, and the diversity of paper topics.

\begin{table*}[!t]
\centering
\setlength{\aboverulesep}{0pt}
\setlength{\belowrulesep}{0pt}
\setlength{\extrarowheight}{1pt}
\setlength{\tabcolsep}{0.2em}
\setlength{\abovecaptionskip}{4pt}
\caption{\method produces qualitatively better results than \mpand. The table shows the top-15 references 
		cited by the paper ``Towards Parameter-Free Data Mining''~\cite{DBLP:conf/kdd/KeoghLR04},
	ordered by the relevance computed with \method (with $ \rho $ set to $ 0.5 $ and $ 0.9 $) and \mpand.
	Papers on the minimum description length principle are highlighted in blue font.}
\makebox[0.4\textwidth][c]{
\begin{tabular}{ p{0.355\textwidth} | p{0.355\textwidth} | p{0.355\textwidth} }
	\toprule
	\makecell[c]{\textbf{\method} ($ \rho=0.5 $)} & \makecell[c]{\textbf{\method} ($ \rho=0.9 $)} & \makecell[c]{\textbf{\mpand}~\cite{DBLP:conf/kdd/TongF06}} \\ 
	\midrule
	On the need for time series data mining benchmarks: a survey and empirical demonstration & On the need for time series data mining benchmarks: a survey and empirical demonstration & On the need for time series data mining benchmarks: a survey and empirical demonstration \\ \midrule
	Making time-series classification more accurate using learned constraints & \highlight{An introduction to Kolmogorov complexity and its applications} & Clustering of time series subsequences is meaningless: implications for previous and future research \\ \midrule
	Distance measures for effective clustering of ARIMA time-series & Making time-series classification more accurate using learned constraints & A symbolic representation of time series, with implications for streaming algorithms \\ \midrule
	Deformable Markov model templates for time-series pattern matching & Distance measures for effective clustering of ARIMA time-series & Making time-series classification more accurate using learned constraints \\ \midrule
	TSA-tree: a wavelet-based approach to improve the efficiency of multi-level surprise and trend queries on time-series data & Deformable Markov model templates for time-series pattern matching & Distance measures for effective clustering of ARIMA time-series \\ \midrule
	Mining the stock market: which measure is best? & TSA-tree: a wavelet-based approach to improve the efficiency of multi-level surprise and trend queries on time-series data & Mining the stock market: which measure is best? \\ \midrule
	Supporting content-based searches on time series via approximation & \highlight{Modeling by shortest data description} & \highlight{Modeling by shortest data description} \\ \midrule
	Clustering of time series subsequences is meaningless: implications for previous and future research & Mining the stock market: which measure is best? & Deformable Markov model templates for time-series pattern matching \\ \midrule
	\highlight{An introduction to Kolmogorov complexity and its applications} & A symbolic representation of time series, with implications for streaming algorithms & Indexing multi-dimensional time-series with support for multiple distance measures \\ \midrule
	A symbolic representation of time series, with implications for streaming algorithms & Clustering of time series subsequences is meaningless: implications for previous and future research & FastMap: a fast algorithm for indexing, data-mining and visualization of traditional and multimedia datasets \\ \midrule
	Indexing multi-dimensional time-series with support for multiple distance measures & \highlight{Inferring decision trees using the minimum description length principle} & TSA-tree: a wavelet-based approach to improve the efficiency of multi-level surprise and trend queries on time-series data \\ \midrule
	\highlight{Modeling by shortest data description} & Supporting content-based searches on time series via approximation & Online novelty detection on temporal sequences \\ \midrule
	\highlight{Inferring decision trees using the minimum description length principle} & \highlight{The similarity metric} & Supporting content-based searches on time series via approximation \\ \midrule
	\highlight{The similarity metric} & Indexing multi-dimensional time-series with support for multiple distance measures & \highlight{Inferring decision trees using the minimum description length principle} \\ \midrule
	Online novelty detection on temporal sequences & Online novelty detection on temporal sequences & Graph-based anomaly detection \\
	\bottomrule
\end{tabular}
}
\vspace{-1em}
\label{tab:result:qualitative2}
\end{table*}

\begin{figure}[!t]
	\centering
	\makebox[\linewidth][c]{\includegraphics[width=0.7\linewidth]{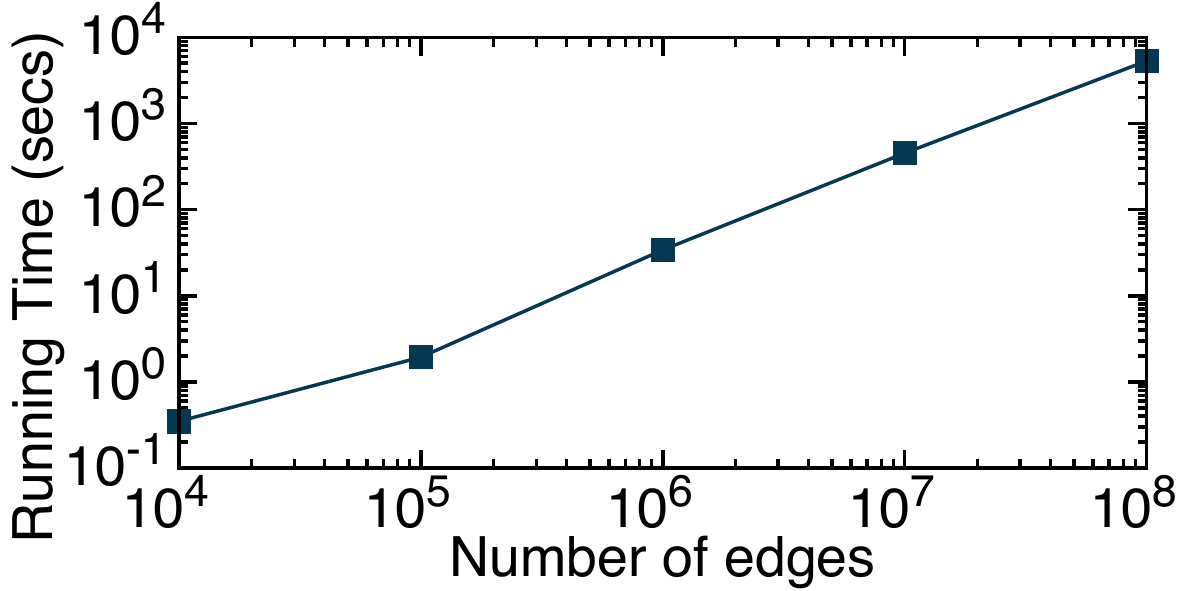}}
	\caption{\method exhibits near-linear scalability.}
	\label{fig:exp:scalability}
\end{figure}

\subsection{\ref{sec:realeval:q2}. Scalability} 

To evaluate the scalability of \method, we created increasingly larger subgraphs of \paperlarge, 
with each subgraph being $10 \times $ larger than the previous one.
\Cref{fig:exp:scalability} reports the running time of \method on these product graphs of varying sizes, 
where the running time was averaged over three simulated users with at least ten products in his $ \setFeedback_u $.
The results show that \method achieves near-linear scalability, successfully scaling up to the largest graph with $ 10^8 $ edges.

\begin{figure}[!t]
	\centering
	\setlength{\abovecaptionskip}{1pt}
	\makebox[\linewidth][c]{\includegraphics[width=0.75\linewidth]{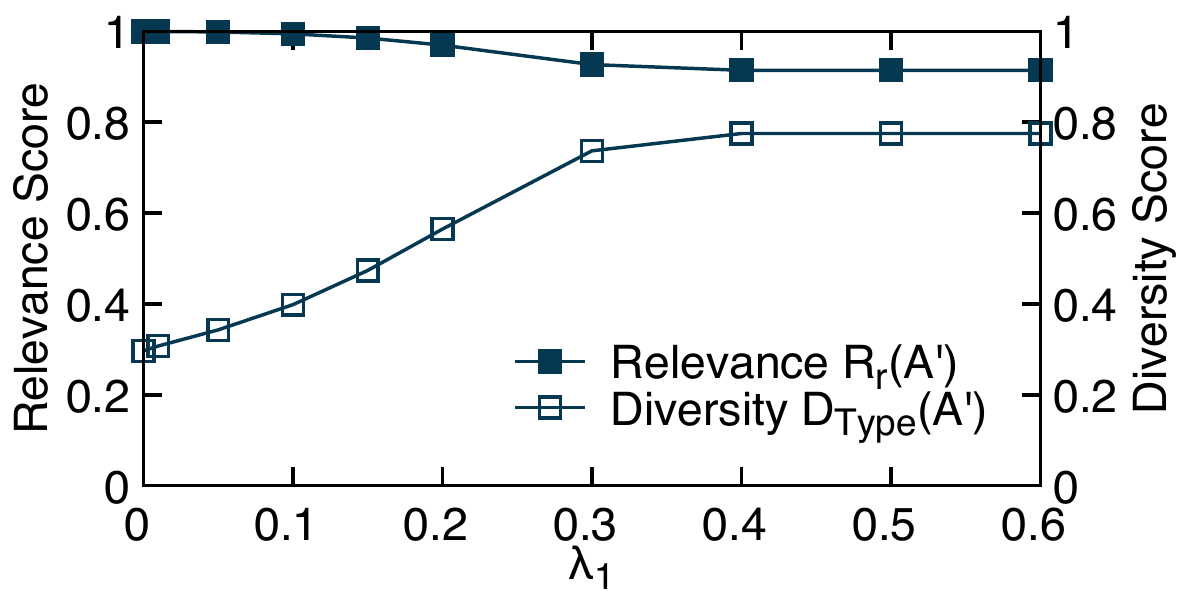}}
	\caption{Relevance-diversity trade-off on \movie. 
		Varying $ \lambda_1 $ affects 
		$ R_r(\setAttribute') $ and $ D_{\text{Type}}(\setAttribute') $
		of the selected \attrs $ \setAttribute' $.}
	\label{fig:exp:sensitivity:diversity}
\end{figure}

\subsection{\ref{sec:realeval:q3}. Relevance-Diversity Trade-Off}

We measure how varying $ \lambda_1 $, the weight for the \attr type diversity, affects 
the relevance score $ R_r(\setAttribute') $ and the diversity score $ D_{\text{Type}}(\setAttribute') $
of the selected \attrs $ \setAttribute' $.
Specifically, we randomly selected 50 users from \movie, who rated at least 10 movies, and 
generated 30 \es varying $ \lambda_1 $ from 0 to 0.6.
In \Cref{fig:exp:sensitivity:diversity}, we report relevance and diversity scores averaged over all users.
\Cref{fig:exp:sensitivity:diversity} shows that while there is a trade-off between relevance and diversity, 
it is possible to attain high relevance and diversity.
As $ \lambda_1 $ is increased from 0 to 0.3, $ D_{\text{Type}}(\setAttribute') $ increases by 148\%,
while $ R_r(\setAttribute') $ decreases by 7\%. 
Results indicate that setting $ \lambda_1 $ to an appropriate value 
can be beneficial in providing diverse and relevant \es to the users.

%% file: 050related.tex
\textbf{Explainable Recommendation.}
Methods for explainable recommendation can be grouped into embedded and post-hoc approaches.
Embedded approaches~\cite{DBLP:conf/www/Wang0FNC18,DBLP:conf/sigir/ChenZ0NLC17,DBLP:conf/kdd/DiaoQWSJW14,DBLP:conf/sigir/ChenCXZ0QZ19,DBLP:conf/sigir/ZhangL0ZLM14}
aim to develop interpretable models, such that explanations for the model decision can be naturally provided.
While embedded methods have high model explainability, different explanation techniques need to be developed for different types of recommendation methods.
Our framework is model-agnostic and can be applied to different recommenders flexibly.
We refer the reader to \cite{DBLP:journals/ftir/ZhangC20} for an in-depth review of embedded methods.

In post-hoc approaches, explanations and recommendations are generated from separate models.
Post-hoc explanations are often 
item-based (e.g., ``Customers who bought this item also bought''~\cite{DBLP:reference/sp/TintarevM15}),
neighbor-based (e.g., ``Your neighbors' rating for this item is''~\cite{herlocker2000explaining}), and
content-based (e.g., keywords~\cite{bilgic2005explaining}, features~\cite{DBLP:conf/recsys/Tintarev07}, 
tags~\cite{DBLP:conf/iui/VigSR09}, or reviews~\cite{DBLP:conf/iui/DonkersL018,DBLP:conf/icdm/WangCYWW018}).
Since these methods typically select an explanation based on one of manually defined templates or generate explanations using one type of data, 
the diversity of their explanations are limited by the form of such templates or the type of input data.
On the other hand, \method is a unified framework that can work with multiple types of data, and 
its diversity increases as we provide more data to the framework.
ExpLOD~\cite{DBLP:conf/recsys/MustoNLGS16} provides more diverse explanations than earlier methods
by using linked open data cloud in a graph-based framework.
However, ExpLOD and others such as \cite{DBLP:conf/iui/VigSR09} produce \es
only using the items in the user profile, ignoring other items and their attributes relevant to the recommendation and the user profile.
By using random walk-based node proximity, \method utilizes both the user profile and other data relevant to it.

\vspace{-1em}
\textbf{Node Importance.}
PageRank (PR)~\cite{page1999pagerank} measures node importance 
by considering the limiting probability of a random surfer that travels over a graph following any out-going edge with uniform probability.
The original PR does not depend on the query, and 
personalized PR (PPR)~\cite{DBLP:conf/www/Haveliwala02} followed PR to estimate query-dependent node importance.
Random walk with restart (RWR)~\cite{DBLP:journals/kais/TongFP08,DBLP:conf/sigmod/JungPSK17} 
can be seen as a special case of PPR that considers one query node.
HITS~\cite{DBLP:journals/jacm/Kleinberg99} first retrieves a focused subgraph with respect to the search query, 
and computes hub and authority scores for each node in the focused subgraph.
SALSA~\cite{DBLP:journals/tois/LempelM01} can be considered as an improvement of HITS,
which also computes hub and authority scores like HITS, while less susceptible to the tightly knit community (TKC) effect than HITS.
HAR~\cite{DBLP:conf/sdm/LiNY12} is a generalization of SALSA that deals with multi-relation data, and 
computes hub and authority scores for objects and relevance scores for relations, with respect to a query input.
GENI~\cite{DBLP:conf/kdd/ParkKDZF19} and MultiImport~\cite{DBLP:conf/kdd/ParkKDZF20} are semi-supervised techniques
to estimate node importance by considering both the graph structure and real-world signals of node popularity.
Among the above methods, query-sensitive ones can be used to measure node-to-node proximity,
which have also been used to identify a subset of nodes or a subgraph,
which have strong connections to the query nodes~\cite{DBLP:conf/kdd/FaloutsosMT04,DBLP:conf/kdd/TongF06} and
are important in connecting source and target nodes~\cite{DBLP:conf/pakdd/TongPFYE10,DBLP:journals/ir/TongPFYE12}.
In this work, we consider the effectiveness of these approaches for the task of recommendation \e,
and use the most effective ones to define the relevance score, which best satisfy the axioms of good \es.

%% file: 060conclusion.tex
In this paper, we present a graph-based formulation of the problem of recommendation \e,
and develop \method, a unified model-agnostic framework which can produce concise, diverse, and personalized \es in a principled manner,
based on various types of product and user data.
We show the effectiveness and efficiency of \method in an evaluation using axioms and real-world data.
In this work, we propose preference retrieval as one way of evaluating the \e quality.
Developing additional automatic and objective evaluation metrics that can measure the quality of \es from different perspectives
will also be an important direction for future research on explainable recommendations.

%% file: 070appendix.tex
\vspace{-0.1em}
\subsection{Experimental Settings}
\label{sec:appendix:settings}
\vspace{-0.1em}

\textit{Machine.}
We ran experiments on a machine with 32 Intel Xeon CPU E7-8837 cores at 2.67GHz, and 1 TB of memory.

\textit{Parameters.}
For \ppr~\cite{DBLP:conf/www/Haveliwala02}, we set the damping factor to $ 0.85 $.
For \har~\cite{DBLP:conf/sdm/LiNY12}, we set the weighting parameters $ \alpha, \beta $, and $ \gamma $ to $ 0.15 $.
For \method, we set $ \rho = 0.5$ and budget B to $ 15 $, unless otherwise stated.
For \explod~\cite{DBLP:conf/recsys/MustoNLGS16}, we followed the settings used in~\cite{DBLP:conf/recsys/MustoNLGS16}, 
where the two weighting factors $ \alpha $ and $ \beta $ were set to $ 0.5 $.